\documentclass{article}
\usepackage{amsmath}
\usepackage{amsfonts}
\usepackage{amsthm}
\usepackage{blkarray}
\usepackage{multirow}
\newcommand{\gF}[1][2]{\mathbb{F}_{#1}}
\newcommand{\cS}{\mathcal{S}}
\newcommand{\cF}{\mathcal{F}}
\newcommand{\cT}{\mathcal{T}}

\newcommand{\RR}{\mathbb{R}}
\newcommand{\NN}{\mathbb{N}}

\newtheorem{definition}{Definition}
\newtheorem{proposition}{Proposition}
\newtheorem{lemma}{Lemma}
\newtheorem{theorem}{Theorem}
\newcommand{\myBig}[1]{\makebox(0,0){\huge $#1$}}
\newcommand{\bcF}{\myBig{\cF}}
\DeclareMathOperator{\supp}{supp}
\DeclareMathOperator{\wt}{wt}
\DeclareMathOperator{\Mat}{Mat}
\title{Binary Non-Tiles}
\author{Don Coppersmith \and Victor~S. Miller}
\date{\today}
\begin{document}
\maketitle
\begin{abstract}
  A subset $V \subseteq \gF[2]^n$ is a {\em tile} if $\gF[2]^n$ can be
  covered by disjoint translates of $V$.  In other words, $V$ is a
  tile if and only if there is a subset $A \subseteq \gF[2]^n$  such
  that $V+A = \gF[2]^n$ uniquely (i.e., $v + a = v' + a'$ implies that
  $v=v'$ and $a=a'$ where $v,v' \in V$ and $a,a' \in A$).  In some
  problems in coding theory and hashing we are given a putative tile
  $V$, and wish to know whether or not it {\em is} a tile.  In this
  paper we give two computational criteria for certifying that $V$
  {\em is not} a tile.  The first involves impossibility of a bin-packing
  problem, and the second involves infeasibility of a linear
  program.  We apply both criteria to a list of putative tiles given
  by Gordon, Miller, and Ostapenko in the context of hashing to find
  close matches, to show that none of them are, in fact, tiles.
\end{abstract}

\section{Tiles}
\label{sec:tiles}

We first define tiles, and make some observations about them.
Discrete tiles arise in problems in perfect codes \cite{EtzionVardy,EtzionVardy2},
and hashing \cite{GordonMillerOstapenko}. Many of their properties
have been extensively analyzed in \cite{TilingBinary}.

In \cite{GordonMillerOstapenko} the
authors give  a table
of ten putative tiles, which is reproduced in 
Table~\ref{optk}.
In
Section 9 of \cite{TilingBinary} the authors give several criteria for
non-tiling.  However, none of these apply to the putative tiles given
in Table~\ref{optk}.  Non-tilings
have also been analyzed in \cite{FullRank8} and \cite{FullRank9}.

\begin{definition}
  Let $U$ be a finite-dimensional vector space over $\gF[2]$.  A
  subset $V \subseteq U$ (which we also refer to as a {\em region}) is
  a {\em tile} (of $U$), if $U$ can be written as the disjoint union
  of translates of $V$.
  In other words, there is a subset $A
  \subseteq U$ such that every element of $U$ can be written {\em
    uniquely} in the form $v + a$, where $v \in V$ and $a \in A$.

  Note that the definition is symmetric in $V$ and $A$, so
  that $A$ is also a tile, called a {\em complement} of $V$.
\end{definition}
\begin{lemma} A subset $V \subseteq U$ of a vector space over $\gF[2]$
  is a tile if and only if
  \begin{eqnarray}
    \label{eq:Tile}
    V+A  & = U \\
    \label{eq:Tile2}
    (V+V) \cap (A+A) & = \{0 \},
  \end{eqnarray}
  where $X+Y := \{x+y : x \in X, y \in Y\}$.
  \end{lemma}
  \begin{proof}  
    The equation
  \eqref{eq:Tile} says that every element of $U$ can be written as a
  sum of an element of $V$ and an element of $A$.  The equation
  \eqref{eq:Tile2} says that this representation is unique.
  Since we're working in characteristic 2, if $v, v' \in V, a,a' \in
  A$ then $v+a = v'+a'$ if and only if $v+v'=a+a'$. Uniqueness of the
  representation is equivalent to both sides being 0.
\end{proof}
\begin{definition}
  If $X \subseteq U$ denote by $\langle X \rangle$ the linear
  subspace generated by $X$ (i.e. the smallest linear subspace
  containing $X$).
\end{definition}
\begin{lemma}
  A subset $V \subseteq U$ is a tile of $U$ 
if and only if $V$ is a tile of
$\langle V \rangle$.
\end{lemma}
\begin{proof}
  Let $W$ be a linear complement of $\langle V \rangle$ (i.e.,
  $W$ is a linear subspace and $W + \langle V \rangle = U$ and $W \cap
  \langle V \rangle = \{0\}$), and $A$ be a complement of $V$ in
  $\langle V \rangle$, then $A+W$ is a complement of $V$.  If $B$ is a
  complement of $V$ in $U$, then $B \cap \langle V \rangle$ is a
  complement of $V$ in $\langle V \rangle$.
\end{proof}
Thus we may
assume, without loss of generality, that $\langle V \rangle =
U$.
\begin{definition}
A subset $V \subseteq U$ is a {\em proper tile} of $U$ if $V$ is a
tile of $U$ and $\langle V \rangle = U$.

If we also assume that the complementary tile $A$ is proper, we say
that the tiling is {\em full-rank}.
\end{definition}

Suppose that we are given a subset $V \subseteq \gF[2]^n$.  We wish to
know whether or not $V$ is a tile, and if it is, to find a complement
$A$.  In the next two sections we give two methods of showing that $V$
is not a tile, and apply them to the list of putative tiles in
Table~\ref{optk} (taken from \cite{GordonMillerOstapenko} using a
slightly different notation), to show that none of them are, in fact, tiles.

In the following we freely use the common practice of identifying a
bit vector with the set of indices in which the vector has the value
1.
The motivating examples we use here are from calculations performed in
\cite{GordonMillerOstapenko} which is concerned with finding hash
functions which are optimal in a particular sense.  One way of
constructing an optimal hash function $H: \gF[2]^n \rightarrow
\gF[2]^r$ would be to find a region $V \subset \gF[2]^n$ with $|V| =
2^{n-r}$ which is both a tile and optimal in the following sense:

\begin{definition}  
  If $S \subseteq \gF[2]^n$, and $p \in (0,1)$, define
  \begin{displaymath}
    f_S(p) := \sum_{a,b \in S} (p/(1-p))^{\wt(a+b)},
  \end{displaymath}
  where $\wt(x)$ is the Hamming weight of the bit vector $x$.
  Say that a region $S \subseteq \gF[2]^n$ is {\em optimal at $p$} if
  \begin{displaymath}
    f_S(p) \ge f_{S'}(p)
  \end{displaymath}
  for all $S' \subseteq \gF[2]^n$, with $|S'| = |S|$.

  We say that $S$ is {\em optimal} if there is a $0 < p < 1$ such that $S$
  is optimal at $p$.
\end{definition}

\begin{definition}
  A {\em set system} is a finite collection of finite subsets of the
  positive integers $\NN$.
\end{definition}
We find it convenient to identify subsets $S \subseteq \gF[2]^n$
with the set system given by the non-zero coordinates of the vectors in $S$:
\begin{displaymath}
  \cS(S) := \{ \{ i \in \NN: v_i = 1\} : v \in S \}.
\end{displaymath}
So in terms of set systems we have polynomials
\begin{displaymath}
  f_{\cS}(p) = \sum_{A,B \in \cS} (p/(1-p))^{|A \Delta B|},
\end{displaymath}
where $A \Delta B$ denotes the symmetric difference.

In order to state the theorem from \cite{GordonMillerOstapenko} we
define a partial order on subsets of $\NN$:
\begin{definition}
  We define a partial order on finite subsets of $\NN$, called the
  {\em right shifted partial order} written $S
  \ge_R T$, as follows:  For a finite subset $S \subseteq \NN$, denote
  by $S_i$ the $i$-th largest element of $S$.  Then we say that $S
  \ge_R T$ if $|S| \ge |T|$ and $S_i \ge T_i$ for $ 1 \le i \le |T|$.
\end{definition}

\begin{definition}
  We say that two set systems, $\cS$ and $\cT$ are {\em isomorphic} if there
  is a permutation $\pi$ of $\NN$ which fixes all but finitely many
  elements, and a finite subset $A \subset \NN$, such that
  $\cT = \{ A \Delta \pi(S) : S \in \cS \}$, where $\Delta$ denotes
  the symmetric difference.

  Note that if $\cS$ is isomorphic to $\cT$ then $f_{\cS}(p) = f_{\cT}(p)$.
\end{definition}

\begin{theorem}[Gordon, Miller, Ostapenko]  An optimal region $\cS$ is 
  isomorphic to an order ideal for the order $\ge_R$.
\end{theorem}

In the table below we give the putative tiles taken from
\cite{GordonMillerOstapenko}.  These were found by using the above
theorem to restrict the search to all order ideals of cardinality 64.
Instead of giving them as collections of bit vectors we specify them
as set systems.  Each is an order ideal for the order $\ge_R$ which we
specify by giving its set of maximal elements.  Each of the putative
tiles has cardinality 64.  The column labeled $k$ indicates that a
putative complement $A$ has cardinality $2^k$.

\begin{table}[h!]
  \caption{Putative tiles}
  \begin{center}
    \begin{tabular}{|c||c|l|} \hline
      $k$ & $n$ & generators of $V$ \\
      \hline\hline
      6 & 12 & $\{11\},\{10,5\}, \{9,8\}$ \\
      7 & 13 & $\{12\},\{10,4\},\{9,8\}$ \\
      8 & 14 & $\{13,2\},\{13,1,0\},\{3,2,0\}$ \\
      9 & 15 & $\{14,1,0\},\{10,2\}$ \\
      16 & 22 & $\{21,1\}$ \\
      17 & 23 & $\{22,0\},\{19,1\}$ \\
      18 & 24 & $\{23,0\},\{17,1\}$ \\
      19 & 25 & $\{24,0\},\{15,1\}$ \\
      20 & 26 & $\{25,0\},\{13,1\}$ \\
      21 & 27 & $\{26,0\},\{11,1\}$ \\
      \hline
    \end{tabular}
    \label{optk}
  \end{center}
\end{table}

For example, the first line in Table~\ref{optk} lists generators
$$\{11\}, \{10,5\}, \{9,8\}.$$

Then $V$ consists of the following 64 sets, 
each $\le_R$  at least one of the generators:

\begin{displaymath}
  \begin{array}{rll}
    \{m\}    , & 0 \leq m \leq 11:    & \mbox{12 sets} \\
    \{10,m\} , & 0 \leq m \leq 5:     & \mbox{6 sets} \\
    \{n,m\}  , & 0 \leq m < n \leq 9: & \mbox{45 sets} \\
    \emptyset : &                     & \mbox{1 set} 
  \end{array}
\end{displaymath}

\section{Bin Packing}
\label{sec:binpacking}

A straightforward combinatorial approach to showing non-tiling is via
bin packing.  This is the idea:  We are given a subset $V \subseteq U$
of an $\gF[2]$ vector space and we wish to prove that it {\em cannot}
be a tile.  We find an auxilliary linear projection 
\begin{displaymath}
  \pi : U \rightarrow W,
\end{displaymath}
where $W$ is another $\gF[2]$ vector space.  The projection $\pi$
partitions the elements of $U$ into equivalence classes: $v \sim v'$
if and only if $\pi(v) = \pi(v')$.  We say that a subset of $U$ of the
form
\begin{displaymath}
  P_{a,w} := \{ v+a : v \in V, \pi(v+a) = w \},
\end{displaymath}
where $a \in U, w \in W$ is a {\em piece}, and that a subset of the
form $\pi^{-1}(w)$ is a {\em bin}.  Note that $P_{a,w} = P_{0,w +
  \pi(a)} + a$.
We then have
\begin{lemma}
  Let $V \subset U$ be a tile with complement $A \subset U$, and $\pi:
  U \rightarrow W$ a linear projection of vector spaces.  Then for all
  $w \in W$ the collection $\{ P_{a,w} : a \in A, P_{a,w} \ne
  \emptyset \}$ is a partition of the bin $\pi^{-1}(w)$.
\end{lemma}
\begin{proof}
  If $P_{a,w} \cap P_{a',w} \ne \emptyset$ then there are $v,v' \in V$
  such that $a+v = a'+v'$.  By definition of a tile this implies that
  $a=a', v= v'$.  We also have $\cup_{a \in A} P_{a,w} = \pi^{-1}(w)$
  by surjectivity of $\pi$
\end{proof}

This approach easily shows non-tiling for the last eight entries in
Table~\ref{optk} by using arguments about the cardinalities of the
pieces $P_{a,w}$.  Note that since $P_{a,w} = a + P_{0,w+\pi(a)}$ that the
possible cardinalities of $P_{a,w}$ are the same as those for $P_{0,w}$
In addition, by the above lemma, every one of
the pieces $P_{0,w} \ne \emptyset$ {\em must} be used in the partition of
some $\pi^{-1}(w)$.  In the arguments below we list the multiset of
cardinalities of the non-empty $P_{0,w}$ for all $w$.

In Table~\ref{proj} below we give the results of projecting each of
the putative tiles in Table~\ref{optk} onto coordinates $r, \dots,
n-1$.  The column labeled ``piece census'' specifies the multiset of
piece sizes obtained.  For example, in the row labeled $k=8$, the
entry 10*5, 1*6, 1*8 means that there are 10 pieces of size 5, 1 of
size 6, and 1 of size 8.  We work out this example in detail to show
the idea:

There are ten pieces of size 5:
\begin{displaymath}
\{ \{m,2\}, \{m,1,0\}, \{m,1\}, \{m,0\}, \{m\} \} \text{ for }
m=4,\dots, 13.
\end{displaymath}
There is one piece of size 6:
\begin{displaymath}
  \{ \{3,2,0\}, \{3,1,0\}, \{3,2\}, \{3,1\}, \{3,0\}, \{3\} \},
\end{displaymath}
and one piece of size 8:
\begin{displaymath}
  \{ \{2,1,0\}, \{2,1\}, \{2,0\}, \{2\}, \{1,0\}, \{1\}, \{0\}, \{\} \}.
\end{displaymath}

For the rows with $k=8$ and $k=9$, the bin size is 8, and the
minimum piece size is $\ge 4$.  Placing a piece of size 5 
 leaves no way to fill up a bin of size 8.  Similarly, in the rows
for $k=16,\dots, 21$,
the bin size is 4, and there are no pieces of size 1.  Placing a
piece of size 3 leaves no way of filling up the bin in which it is placed.

Thus, none of the last eight rows in Table~\ref{optk} can be a tile.

We have not been able to find a projection to make this argument work
for the first two rows of the table.

\begin{table}[htbp]
   \caption{Projections}
 \begin{center}
  \begin{tabular}[h]{|r|r|r||c|}
    \hline
    $k$ & $r$ & bin size &piece census \\
    \hline \hline
    8 & 3 & 8 & 10*5, 1*6, 1*8 \\
    9 & 3 & 8 & 4*4, 8*5, 1*8 \\
    16 & 2 & 4 & 20*3,  1*4 \\
    17 & 2 & 4 &3*2, 18*3, 1*4 \\
    18 & 2 & 4 & 6*2, 16*3, 1*4\\
    19 & 2 & 4 & 9*2, 14*3, 1*4 \\
    20 & 2 & 4 & 12*2, 12*3, 1*4 \\
    21 & 2 & 4& 15*2, 10*3, 1*4 \\
    \hline
  \end{tabular}
\label{proj}
\end{center}
\end{table}

\section{Linear Programming}
\label{sec:LP}

In this section we'll rewrite the defining conditions for a tile in
terms of a linear program.  We'll identify subsets $S \subseteq U$ with
their characteristic functions $\chi_S : U \rightarrow \RR$:
$\chi_S(x) = 1$ if $x \in S$ and 0 otherwise.  Denote convolution
of functions $f,g : U \rightarrow \RR$ by
\begin{displaymath}
  f \star g (x) = \sum_{y \in U} f(y) g(x+y),
\end{displaymath}
and the Fourier transform
\begin{displaymath}
  \widehat{f}(y) = \sum_{x \in U} (-1)^{x\cdot y} f(x).
\end{displaymath}
As is well known:
\begin{displaymath}
  \widehat{f \star g} = \widehat{f} \, \widehat{g}.
\end{displaymath}
Note that if $X,Y \subseteq U$ then
\begin{displaymath}
  \chi_X\star \chi_Y(z) := | \{ (x,y) : x \in X, y \in Y, x+y = z \}|,
\end{displaymath}
the number of ways of writing $z$ as the sum of an element in $X$ and
an element of $Y$.
Thus, we may express the
condition for $V$ to be a tile (with $A$ as a complement) as
\begin{equation}
  \label{eq:rewrite1}
  \chi_V \star \chi_A = 1,
\end{equation}
where 1 denotes the constant function with value 1.  Although this is
a necessary and sufficient condition (along with the condition that
$\chi_A(u) \in \{0,1\}$) it proves to be too weak to use as a linear
programming criterion to certify non-tiling.  We supplement it with
the condition derived from
\eqref{eq:Tile2}:
\begin{equation}
  \label{eq:rewrite2}
  (\chi_V \star \chi_V)(\chi_A \star \chi_A) = |U| \delta,
\end{equation}
where $\delta : U \rightarrow \RR$ is the function $\delta(0) = 1$ and
$\delta(x) = 0$ when $x \ne 0$.

Taking the Fourier transform of \eqref{eq:rewrite1} yields
\begin{equation}
  \label{eq:rewritea}
  \widehat{\chi_V} \widehat{\chi_A} = |U| \delta.
\end{equation}

This suggests a linear program in which we use the variables $\chi_A
\star \chi_A(u)$ instead of $\chi_A(u)$.  We are given $V$.  Since $|V| |A| =
2^n$ we know $|A|$ (if it exists). We have variables $b_u$ for $u \in
\gF[2]^n$ and $c_x$ for $x \in \gF[2]^n$. We'll want
\begin{displaymath}
  b_u = \chi_A \star \chi_A(u),
\end{displaymath}
and
\begin{displaymath}
  c_x = |\widehat{\chi_A}(x)|^2.
\end{displaymath}
We have the conditions
\begin{align}
  0 \le b_u \le & |A|  \text{ and is an integer,}
  \label{eq:lp} \\
  0 \le c_x \le & |A|^2  \text{ and is the square of an integer,}
  \tag{\ref{eq:lp}a}\\
    b_0 = & |A| \tag{\ref{eq:lp}b}\\
      c_0 = & |A|^2 \tag{\ref{eq:lp}c}\\
        b_u  = & 0 
        \text{ if } u \ne 0 \text{ and } \chi_V \star \chi_V(u)  \ne  0  \tag{\ref{eq:lp}d}\\
         c_x  = & 0  \text{ if } x \ne 0 \text{  and } \widehat{\chi_V}(x) \ne 0 \tag{\ref{eq:lp}e}\\
    c_x  = & \sum_u (-1)^{x \cdot u} b_u  \text{ for all }x \tag{\ref{eq:lp}f}
\end{align}

If we drop the conditions about $b_u$ being an integer and $c_x$ being
the square of an integer, we get a linear program which must be
feasible if $V$ is a tile.

One problem with this linear program is that it has a large number of
nonzero coefficients.  Just the condition that the $c_x$ be the
Fourier transform of the $b_u$ yields
$2^n(2^n-\max(|\supp(\chi_V \star \chi_V)|,|\supp(\widehat{\chi_V})|)$
nonzero coefficients (we can either write the $c_x$ as the transform
of the $b_u$ or the $b_u$ as the inverse transform of the $c_x$,
whichever yields a smaller system).  We can immediately halve the
number of nonzero coefficients by adding the equation

\begin{displaymath}
  \sum_u b_u = |A|^2
\end{displaymath}
to the remaining equations for the Fourier tranform.  However, the
number of nonzeros is still quite large.  We can greatly reduce this
by means of ideas from the fast Fourier transform.  We create new variables
corresponding to the intermediate results of the transform.  The usual
sort of bookeeping now yields $3n 2^n$ nonzeros (since each
``butterfly'' involves 3 variables) and introduces $n 2^n$ new
variables.

Here are the details:
Introduce variables $t_{i,j}$ with $0 \le i \le n, 0 \le j < 2^n$,
with $t_{0,j} = c_j$.  For $0 \le i < n$, $0 \le k < 2^i$, and $ 0 \le
j < 2^{n-i-1}$, introduce the equations:
\begin{align}
  t_{i+1,j + 2^{n-i}k} & = t_{i,j + 2^{n-i}k} + t_{i,j + 2^{n-i}k +
    2^{n-i-1}}   \label{eq:butterfly} \\
  t_{i+1,j + 2^{n-i}k + 2^{n-i-1}} & = t_{i,j + 2^{n-i}k} - t_{i,j +
    2^{n-i}k + 2^{n-i-1}}
  \tag{\ref{eq:butterfly}a}
\end{align}
The values $t_{n,j}$ are the values of the Fourier transform of $c_j$.
We can also achieve a significant savings in our problem by noting
that whenever one of the variables on the right hand sides of
\eqref{eq:butterfly} is 0 (which is the case for a significant
fraction of the $c_j$), then we can ``pass through'' the remaining
variable, or its negation, or a 0 if both are 0, and not create a new variable.  We
note the
effect of this special case in Table~\ref{results}, by comparing $n 2^n$
to the actual number of variables needed.

One nice feature of the approach using linear programming is that the
conditions like full-rank for the complementary tile $A$ can be
described as linear inequalities.
\begin{proposition}
  A subset $A \subseteq U$ containing $0$ generates $U$ as a linear subspace if and
  only if
  \begin{displaymath}
    |\widehat{\chi_A}(x)| \le |A|-2,
  \end{displaymath}
  for all $0 \ne x \in U$.
\end{proposition}
\begin{proof}
  The value $\widehat{\chi_A}(x)$ is the sum of $|A|$ terms each of which is
  $\pm 1$. Thus $\widehat{\chi_A}(x) = |A|$ if and only if $x \cdot a = 0$ for
  all $a \in A$.  This can happen if and only if $A$ does not generate
  $U$.  Similarly $\widehat{\chi_A}(x) = -|A|$ if and only if $x \cdot a = 1$
  for all $a \in A$.  This is impossible since $ 0 \in A$.  Further
  note that $\widehat{\chi_A}(x) \equiv |A| \bmod 2$, thus the value of $|A|-1$
  is impossible for $|\widehat{\chi_A}(x)|$.
\end{proof}
Thus, if we use variables representing $|\widehat{\chi_A}(x)|^2$ we can express
full rank as $|\widehat{\chi_A}(x)|^2 \le (|A|-2)^2$.

For the first four of the ten putative tiles given in Table~\ref{optk}
the resulting linear programming problem was small enough so that either
{\tt glpsol}\footnote{{\tt glpsol} is the standalone solver contained
  in {\tt GLPK} -- the GNU Linear Programming Kit \texttt{http://www.gnu.org/software/glpk}}
 or
 {\tt CPLEX}\cite{CPLEX}
 could handle it.  The results of this approach applied to the first four
cases in Table~\ref{optk} is given in Table~\ref{results}.

\begin{table}[htbp]
  \caption{Results from {\tt CPLEX}}
  \begin{center}
  \begin{tabular}[h]{|r|r|r||r|r|r|r|}
    \hline
    $k$ & $n$ & $n 2^n$ & time in seconds & rows & variables & nonzeros\\
    \hline \hline
    6 & 12 & 49152 & 0.46 & 33569 & 33414 & 99465 \\
    7 & 13 & 106496 & 72.60 &  74349 & 74710 & 221693 \\
    8 & 14 & 229376 & 0.54 & 140312 & 142632 & 419864 \\
    9 & 15 & 491520 & 39.72  &  321016 & 327828 & 961832 \\
    \hline
  \end{tabular}
  \label{results}
\end{center}
\end{table}
The system of linear equations for $k=8$ were inconsistent.  The
linear programming problems for the last six rows of the table were too
large for {\tt CPLEX} to handle.

\subsection{Certificates of non-tiling from Linear Programming}
\label{sec:farkas}

For the use of the linear programming criterion, we use the following
well-known criterion to produce a certificate of non-tiling.
\begin{lemma}[Farkas]
  Let $A \in \Mat_{m \times n}(\RR), b \in \RR^n$.  Then exactly one
  of the following statements is true:
  \begin{enumerate}
  \item There is an $y \in \RR^n$ such that $Ay = b$ and $y \ge 0$.
  \item There is a $z \in \RR^m$ such that $A^{T} z \ge 0$ and $b^T z
    < 0$.
  \end{enumerate}
\end{lemma}
We explicity describe the matrix $A$ and vector $b$ for our problem.
A solution to the second alternative above will consistute a
certificate of non-tiling.  Note that since the linear programming
problem in second alternative is homogeneous, in practice we replace
it by
\begin{quotation}
  There is a $z \in \RR^m$ such that $A^{T} z \ge 0$ and $b^T z = -c$
  for some convenient $c > 0$.
\end{quotation}

For convenience we write $f(x) = \chi_A \star \chi_A(x)$ The rows and
columns of $A$ are divided into {\em regions}.  The first region (1)
of rows is indexed by elements $\widehat{x} \in \widehat{U}$.  The
next region of rows, which is indexed by the symbols $f(x)$ for $x \in
U$ corresponds to $x \in U$ for which $f(x)$ has a fixed value --
either $|A|$ for $x = 0$ (region (2)), or $0$ for those $x$ with
$\chi_V(x) \star \chi_V(x) \ne 0$ (region (4)).  The next region is
similar to the last, is indexed by symbols of the form
$\widehat{f}(\widehat{x})$, for which $\widehat{f}(\widehat{x})^2$ has
a fixed value -- either $|A|^2$ (region (3)) for $\widehat{x} =
\widehat{0}$ or $0$ for those $\widehat{x}$ for which
$\widehat{\chi_V}(\widehat{x}) \ne 0$ (region (5)).

The columns are divided into two regions -- the first indexed by the
elements $f(x)$ and the second by elements $\widehat{f}(\widehat{x})$.
\begin{equation}
  \label{eq:dual:problem}
  \myBig{$A=$} \quad \quad \quad
  \begin{blockarray}{cccc|ccccc} 
      &      & f(x) &      &       &\widehat{f}(\widehat{x}) \\
    \begin{block}{(cccc|cccc)c}
      &      &      &      & -1 \\
      & \bcF &      &      &       & -1  &        &        & \text{(1)}\\
      &      &      &      &       &     & \ddots &        &  \\
      &      &      &      &       &     &        & -1 \\
      \cline{1-8}
    1&       &      &      &       &     &        &        &     \text{(2)} \\
    \cline{1-8}
     &       &      &      &       & 1   &        &        &     \text{(3)} \\
    \cline{1-8}
     & \ldots& 1    &      &       &     &        &        &      \text{(4)} \\      
    \cline{1-8}
     &       &      &      &       & \ldots&     1 &      &      \text{(5)} \\      
    \end{block}
  \end{blockarray}
\end{equation}

The vector $b$ is 0 everywhere, except for $b_{f(0)} = |A|,
b_{\widehat{f}({\widehat{0}})} = |A|^2$.

We denote elements of $C_2^n$ (the $n$-fold product of the group of
order 2) by $x$, and elements of its dual, $\widehat{C_2^n}$ by
$\widehat{x}$.  The vector $z$ giving the Farkas certificate is
denoted by $z$.  Note that we've eliminated the upper bounds in
equation \eqref{eq:lp}.

The part of the matrix denoted by $\cF$ is the matrix of Fourier
transform on $C_2^n$.  We can use the same idea as in
\eqref{eq:butterfly} to make the resulting set of equations much more
sparse at the expense of introducting auxilliary variables.

In terms of equations these are:

For $b^T z < 0$:
\begin{equation}
  \label{eq:rhs}
  z_{f(0)} |A| + z_{\widehat{f}({\widehat{0}})} |A|^2 < 0.
\end{equation}
As remarked above we set the right hand side of this to any convenient
negative constant.
Column $f(0)$:
\begin{equation}
  \label{eq:region:1}
  \sum_{\widehat{x} \in \widehat{C_2^n}} z_{\widehat{x}} + z_{f(0)}
  \ge 0
\end{equation}
Column $f(x)$ where $ x \in V + V$:
\begin{equation}
  \label{eq:region:2}
  \sum_{\widehat{x} \in \widehat{C_2^n}} \widehat{x}(x)
  z_{\widehat{x}} + z_{f(x)} \ge 0.
\end{equation}
Column $f(x)$ where $ x \not \in V + V$:
\begin{equation}
  \label{eq:region:3}
  \sum_{\widehat{x} \in \widehat{C_2^n}} \widehat{x}(x)
  z_{\widehat{x}} \ge 0.
\end{equation}
Column $\widehat{f}({\widehat{0}})$:
\begin{equation}
  \label{eq:region:4}
  -z_{\widehat{0}} + z_{\widehat{f}(\widehat{0})} \ge 0.
\end{equation}
Column $\widehat{f}({\widehat{x}})$, where
$\widehat{\chi_V}(\widehat{x}) = 0$:
\begin{equation}
  \label{eq:region:5}
  -z_{\widehat{x}} + z_{\widehat{f}(\widehat{x})} \ge 0.
\end{equation}
Column $\widehat{f}({\widehat{x}})$, where
$\widehat{\chi_V}(\widehat{x}) \ne 0$:
\begin{equation}
  \label{eq:region:6}
  -z_{\widehat{x}} \ge 0.
\end{equation}
In practice, we found that we could require that inequalities
\eqref{eq:region:4}, \eqref{eq:region:5} and \eqref{eq:region:6} be
equalities, and we could still find a certificate of non-tiling.  This
effectively eliminates the parts of $z$ indexed by
${\widehat{x}}$.

For the first and third rows of Table~\ref{optk} the linear program
solver CLP\footnote{The linear program solver included in the COIN-OR
  package \texttt{http://www.coin-or.org}.  See \cite{COIN:OR}} found
particular easy to describe certificates.  For the second and fourth
rows the certificates did not have an apparent structure which allows
a compact description.

The certificates below are vectors which have 0.0 everywhere
except in the coordinates specified in the tables below.
The part above the line corresponds to indices of the form
$\widehat{f}({\widehat{x}})$ and the part below to indices of the form
$f(x)$.  An expression like $(c,l)$ means a
segment of all
coordinates starting at $c$ of length $l$.  An
expression like $(c,l,s)$ has a similar meaning except that there is a
stride of $s$,  meaning the arithmetic
progression $c, c+s, \dots, c+ms$ where $m$ is the largest integer
such that $c+ms < c + l$.
\begin{table}[h]
  \centering
  \begin{tabular}[h]{||c|c||}
    \hline 
    \hline
    0 & -1024.0 \\
    320 & 1024.0 \\
    640 & 1024.0 \\
    \hline
    (192,64) & 1.0 \\
    (384, 64) & 1.0 \\
    (576, 64) & 1.0 \\
    (768, 64) & 1.0 \\
    (1216, 64)&  1.0 \\
    (1408, 64) & 1.0 \\
    (1600, 64)& 1.0 \\
    (1792, 64)& 1.0 \\
    (2240, 64)& 1.0 \\
    (2432, 64)& 1.0 \\
    (2624, 64)& 1.0 \\
    (2816, 64)& 1.0 \\
    (3264, 64)& 1.0 \\
    (3456, 64)& 1.0 \\
    (3648, 64)& 1.0 \\
    (3840, 64)& 1.0 \\
    \hline
    \hline
  \end{tabular}
  \caption{Certificate of non-tileability for the first putative tile}
  \label{tab:first}
\end{table}
\begin{table}[h]
  \centering
  \begin{tabular}[h]{||c|c||}
    \hline
    \hline
    0 & -8192.0 \\
    \hline
    (0,16384,2) & 1.0 \\
    \hline \hline
  \end{tabular}
  \caption{Certificate of non-tileability for the third putative tile}
  \label{tab:first}
\end{table}

How widely applicable are these methods?  We
have tried the linear programming method on all of the examples of
non-tiles given in \cite{TilingBinary}, and in all cases it has found
a certificate of non-tileability.  In \cite{FullRank8,FullRank9} the
authors show that there are no full-rank tilings in dimension 8 and 9
as a result of the execution of a very long running computer program.
It would be interesting to see if the linear programming method could
supply certficates of non-tileability for the all the cases examined.
\nocite{COIN:OR}
\bibliographystyle{plain}
\bibliography{CDownSets}
\end{document}